\documentclass[11pt]{article}
\usepackage{url}
\usepackage{caption}
\usepackage{tabularx}
\usepackage{booktabs}
\usepackage[pdfstartview=FitH,pdfpagemode=UseNone,backref,colorlinks=true,citecolor=blue,linkcolor=blue]{hyperref}
\usepackage{multirow}
\usepackage{amsfonts}
\usepackage{latexsym,amssymb,amsmath,amsthm,mathrsfs}
\usepackage{arydshln}

\usepackage[ruled,vlined]{algorithm2e}
\usepackage{algorithmic}
\usepackage{cite}
\usepackage[margin=1.1in]{geometry}

\newtheorem{theorem}{Theorem}[section]

\newtheorem{lemma}[theorem]{Lemma}

\newtheorem{corollary}[theorem]{Corollary}

\newtheorem{claim}[theorem]{Claim}
\newtheorem{obs}[theorem]{Observation}

\newtheorem*{mainlemma}{Lemma \ref{lemma:main}}
\newtheorem*{maintheorem}{Theorem \ref{thm:main}}

\newcommand{\eps}{\epsilon}

\newcommand{\C}{\mathcal{C}}
\newcommand{\R}{\mathcal{R}}

\newcommand{\poly}{\operatorname{poly}(n)}
\newcommand{\red}{\mathtt{red}}
\newcommand{\blue}{\mathtt{blue}}
\newcommand{\green}{\mathtt{green}}

\newcommand{\mc}[1]{\mathcal{M}\mathcal{C} #1}

\newcommand{\states}{\Phi}

\newcommand{\length}{\ell}
\newcommand{\paths}{\mathcal{P}}

\newcommand{\ignore}[1]{}

\newcommand{\mvs}{\operatorname{MVS}}
\newcommand{\vsn}{\operatorname{VSN}}

\newcommand{\ca}{\alpha}
\newcommand{\cb}{\beta}

\newcommand{\sigsig}{\sigma, \sigma'}
\newcommand{\cc}{\ca,\cb}
\newcommand{\tw}{\operatorname{tw}}
\newcommand{\pw}{\operatorname{pw}}
\newcommand{\cp}{\mathtt{pairs}}

\newcommand{\Q}{QS}
\newcommand{\mono}{MCE(G)}
\newcommand{\lgg}{L}

	\title{Randomly  coloring graphs of bounded treewidth}
		\author{ Shai Vardi%
		\thanks{California Institute of Technology, Pasadena, CA, 91125, USA. E-mail: {\tt  svardi@caltech.edu}. }}
\begin{document}
		
\maketitle

 \begin{abstract}  We consider the problem of sampling a proper  $k$-coloring of a graph of maximal degree $\Delta$ uniformly at random. We describe a new Markov chain for sampling colorings, and show that it mixes rapidly on graphs of bounded treewidth if $k\geq(1+\eps)\Delta$, for any $\eps>0$.
\end{abstract}

\section{Introduction}

A (proper) $k$-coloring of a graph $G=(V,E)$ is an assignment $\sigma:V \rightarrow \{1, \ldots, k\}$ such that neighboring vertices have different colors. We consider the problem of sampling (almost) uniformly at random from the space of all $k$-colorings of a graph.\footnote{We define precisely what we mean by "almost" in Section~\ref{sec:markovprelims}.} The problem has received  considerable attention from the computer science community in recent years, 
e.g.,~\cite{Jerrum95,Dyer10,Vigoda99,HayesV03,Goldberg10,Tetali10,Lu17}. It also has applications in Combinatorics (e.g.,~\cite{Cayley}) and Statistical Physics (e.g.,~\cite{Stat}). 

Sampling colorings (as well as other combinatorial objects, e.g.,~\cite{JerrumS89,MorrisS04,Brightwell}) is commonly done using  Markov Chain Monte Carlo (MCMC) methods.  
A large body of work on sampling colorings is devoted to analyzing a particular Markov chain,  known as \emph{Glauber dynamics}: Choose a vertex $v$ uniformly at random; choose a color $c$ uniformly at random from the set of available colors (the complement of the set of colors of the neighbors of $v$); recolor $v$ with $c$.
Jerrum~\cite{Jerrum95} showed that the Glauber dynamics mix in time $O(n \log{n})$ when $k>2\Delta$, where $\Delta$ is the maximal degree of the graph.
Vigoda~\cite{Vigoda99} improved the bound on the number of colors to $k> 11\Delta/6$ using a different Markov chain, and showed that it mixes in time $O(nk\log{n})$.  
This remains the best known bound on $k$ for general graphs. A major open question is for what values of $k$ can we sample colors efficiently (i.e., in polynomial time)? It is conjectured (e.g.,~\cite{survey}) that $k=\Delta+2$ colors suffice, and furthermore, that the Glauber dynamics mix rapidly for any $k \geq \Delta+2$.

A lot of  work has  focused on improving the bounds of Vigoda on restricted families of graphs. Dyer and Frieze~\cite{Dyer03} showed that if the maximal degree and girth are $\Omega(\log{n})$, the Glauber dynamics mix in $O(n\log{n})$ for $k>\alpha\Delta$, where $\alpha \approx 1.763$. The degree and girth requirements and the value of $\alpha$ were improved in a  line of works~\cite{HayesV03,DyerFHV13,Hayes03,HayesV05,Molloy04}; see Table~\ref{table1} for a comparison and summary of some milestones.  The current state of the art results exhibit a tradeoff between the value of $\alpha$ and the degree and girth requirements. Hayes and Vigoda~\cite{HayesV03} showed that  on graphs with $\Delta=\Omega{(\log{n})}$ and girth at least $9$, $(1+\eps)\Delta$ colors suffice to ensure fast mixing. On the other hand, the Glauber dynamics have been shown to mix rapidly  on graphs with girth at least $5$ (resp. $6$) and $\Delta>\Delta_0$ (where $\Delta_0$ is some absolute constant) using roughly $1.763\Delta$ (resp. $1.489\Delta$) colors~\cite{DyerFHV13}. Thus far, stronger bounds have only been shown on highly specialized families of graphs, typically with chromatic number smaller than the maximal degree, such as trees~\cite{Martinelli}, planar graphs~\cite{Hayes07}, Erd\"{o}s-R\'{e}nyi graphs~\cite{Dyer10} and cubic graphs~\cite{Lu17}.

\subsection{Results}

Our main result is  an algorithm that efficiently  samples a $((1+\eps)\Delta
)$-coloring  (almost) uniformly at random if the input graph has  logarithmically bounded pathwidth,  for any $\eps>0$.\footnote{Assuming $(1+\eps)\Delta \geq \Delta+2$.}

\begin{theorem}\label{thm:main} (Informal)
	Let $\eps>0$ and $G$ be a graph with maximal degree $\Delta$ and  pathwidth bounded by $O(\log{n})$. There exists a polynomial time algorithm for sampling a $((1+\eps)\Delta)$-proper coloring of $G$ (almost) uniformly at random. 
\end{theorem}
Using the fact that the pathwidth of a graph is at most $O(\log{n})$ times its treewidth~\cite{Korach}, we have the following corollary.  
\begin{corollary}\label{cor:tr} (Informal)
	Let $\eps>0$ and $G$ be a graph with maximal degree $\Delta$ and  treewidth bounded by a constant. There exists a polynomial time algorithm for sampling a $((1+\eps)\Delta)$-proper coloring of $G$ (almost) uniformly at random. 
\end{corollary}

These results improve upon existing results for a large family of graphs. Previously, the best bounds for many graphs of bounded treewidth (with the notable exception of trees) were the results of Hayes and Vigoda~\cite{HayesV03} and  Dyer et al.~\cite{DyerFHV13}. 
We  remove all restrictions on minimal girth or maximal degree and show  fast mixing for sampling a  $(1+\eps)\Delta$ coloring on \emph{any} graph of bounded treewidth. We note that our result concerning pathwidth (Theorem~\ref{thm:main}) is  strictly stronger --  it holds for the entire family of graphs with pathwidth $O(\log{n})$, which includes many graphs with treewidth $\omega(1)$ -- but we highlight the second result because treewidth is a more popular measure: graphs of bounded treewidth have been studied extensively in the past few decades, e.g.,~\cite{Lok,Andrzejak,Loksh,Fomin,Arn,Bruno,KargerS01,Gupta} and are common~\cite{Thorup,Bod98}. Bounded treewidth graphs include  series-parallel graphs, outerplanar graphs,  many control flow graphs and expert systems, Apollonian networks and Halin graphs among others.

\begin{center}
	\renewcommand{\arraystretch}{1.1}
	
	\begin{table}
		\small
		\caption{Comparison of results on sampling $k$-colorings.} 
		
		\label{table1}
		\begin{tabular}{l l l l l   l l l }
			\hline
			Degree & Girth & Graph family		& $k>$    		   &Dynamics		& Mixing time   & Reference \\ \hline
			any	& any&	any   & $2\Delta$ &  Glauber & $O(n\log{n})$ & \cite{Jerrum95}\\
			any& any&	any  & $(1.833\ldots)\Delta$ &  Flip & $O(n\log{n})$ & \cite{Vigoda99}\\
			$\Omega{(\log{n})}$ & $\Omega{(\log{n})}$&		any & $(1.763\ldots)\Delta$ &Glauber & $O(n\log{n})$ & \cite{Dyer03}\\
			$ \Omega{(\log{n})}$		&$ \geq 9$& any & $(1
			+\eps)\Delta$ &Glauber & $O(n\log{n})$ & \cite{HayesV03}\\
			$\geq \Delta_0 \ ^\dagger$	&$\geq 5 $&	any
			& $(1.763\ldots)\Delta$ & Glauber & $O(n\log{n})$ &\cite{DyerFHV13}\\
			$\geq \Delta_0 \ ^\dagger$	&$\geq 6 $&	any& $(1.489\ldots)\Delta$ &Glauber &  $O(n\log{n})$& \cite{DyerFHV13}\\
			$O(1)$	&$\infty$  &  Trees & $4$  &Glauber & $\poly$ & \cite{Lucier11}\\
			$\geq \Delta_0 \ ^\dagger$ & any  &  Planar & $\Omega(\Delta/\log{\Delta})$  &Glauber & $O(n^3\log^9{n})$ & \cite{HayesVV15}\\
			\hdashline
			any&any&treewidth $=O(1)$  & $(1+
			\eps)\Delta$		 & Single-Flaw & $\poly$ & Here\\
			any 	&any&pathwidth $=O(\log{n})$  & $(1+
			\eps)\Delta$		 & Single-Flaw & $\poly$ & Here\\
			\hline
			
		\end{tabular}
		\\[10pt]
		\caption*{$^\dagger$ $\Delta_0$ is some absolute constant. $^\star$ $d$ is a constant.}
	\end{table}
\end{center}

\subsection{Techniques}
The two main methods of bounding the mixing time of random walks are \emph{coupling} and bounding the \emph{spectral gap} of the transition matrix~\cite{Guru16};
methods of bounding the conductance or congestion (which are used to bound the spectral gap) are generally considered to be stronger than coupling methods~\cite{Guru16,Kumar01}. Despite this, most of the work on sampling colorings uses various coupling methods~\cite{Jerrum95,Dyer10,Dyer03,Vigoda99,HayesV03,DyerFHV13,HayesVV15,Hayes07,HayesV05,Molloy04,Hayes03,DyerFFV06}; the Glauber dynamics do not lend themselves easily to the techniques that are usually used for bounding the spectral gap. In particular, bounding the  \emph{congestion}\ of the underlying graph of the transition matrix\footnote{In the underlying graph of the transition matrix $A$ of a Markov chain, the states are represented by vertices, there is an edge $(u,v)$ between two states $u,v \in \Omega$ with weight $w_e=A_{u,v}$ if $A_{u,v} \neq 0$.}  typically involves defining flows between states in the underlying graph. If one can describe a flow in the underlying graph such that the congestion of each edge is not too large, it implies that the spectral gap is large and the Markov chain mixes rapidly~\cite{Sinclair92,Diaconis91}. It is not clear how to construct such flows for the Glauber dynamics, as a transition involves changing a vertex's color to one of its available colors: if all of the neighbors of   a vertex $v$ that is colored $\blue$ are colored $\red$, how do we go about  changing $v$'s color to $\red$?    

In order to facilitate the use of these more powerful techniques, we introduce a new Markov chain, which we call \emph{
	Single-Flaw dynamics}. The difference between the Single-Flaw and  Glauber dynamics is that the Single-Flaw dynamics also allow colorings that have a ``single flaw'' --  there is at least one monochromatic edge, and all monochromatic edges share a vertex. In other words, the coloring is not proper, but there is a single vertex $v$ such that we can reach a proper coloring by changing $v$'s color only. We call such colorings  \emph{singly-flawed}. Concretely, the Single-Flaw dynamics Markov chain is the following: choose a vertex $v$ and a  color $c$ at random. If changing $v$'s color to $c$ results in a coloring that is either proper or singly-flawed, change $v$'s color to $c$. Otherwise do not.

The main advantage afforded by this Markov chain is that it   allows us to define simple \emph{canonical paths} between two states (colorings) $\ca$ and $\cb$:  select some order on the vertices, $v_1, \ldots, v_n$. Starting from $v_1$, for each vertex $v_i$, change its color to $\cb(v_i)$. If the transition leads to a proper coloring, continue to $v_{i+1}$. Otherwise, ``fix'' the monochromatic edges by recoloring the neighbors of $v_i$ that are also colored $\cb(v_i)$ (as $k \geq \Delta+2$ there is always at least one available color). When there are no  monochromatic edges remaining, continue to vertex $v_{i+1}$.  

We first describe a simple attempt to adapt the canonical paths argument  of Jerrum and Sinclair~\cite{JerrumS89,JerrumSV04} to our setting, using the canonical paths described above. Although it fails in all non-trivial cases, it is instructive as it exemplifies an important part of our method. For every edge in the underlying graph, (try to) describe an injective function from paths going through the edge to the state space of the chain. If we can describe such a function, it would mean that at most $|\Omega|$ paths use each edge. If there were no ``fixing phase'' (i.e.,  the graph was disconnected), this would be easy: Let $t$ be the transition from the state  $\sigma$ such that   the color of $j^{th}$ vertex is changed to $c$. The injective function would map the path from $\ca$ to $\cb$  to the following coloring $\sigma^*$: for vertices $v_i : i = 1,2,\ldots j$,  $\sigma^* = \ca(v_i)$, for all other vertices $v_i : i=j+1,\ldots, n$,  $\sigma^*(v_i)=\cb(v_i)$. The mapping is injective because (i) $\sigma^*$ a proper coloring and (ii) knowing $t$ and $\sigma^*$ allows us to recover $\ca$ and $\cb$, as  
\begin{align*}\sigma(v_i) =
\begin{cases}
\beta(v_i) & i= 1,2,\ldots, j\\
\alpha(v_i) & \text{ otherwise.}
\end{cases}
\end{align*}
This implies that at most $|\Omega|$ paths would  use each transition, if there were no fixing stage. We note that it is not necessary to show that \emph{at most} $|\Omega|$ paths use each transition to show polynomial time mixing; it suffices to show that $|\Omega|\cdot\poly$ paths use each transition~\cite{JerrumS89}.

In our case, however, there \emph{is} a fixing stage, and this complicates matters considerably. 
The following toy example serves to  highlight some of the challenges. Assume $G=(V,E)$ is a star with $|V|=n$ (it is actually easy to sample colorings of stars~\cite{Martinelli,Lucier11}; the reader may want to think of $G$ as a subgraph of some graph $G'$). Denote the center of the star by $v$, and assume that the canonical paths are such that $v$ is first and some vertex $w$ is last. If we need to fix the color of a leaf $u$ that is currently colored $\red$ (because $v$ was recolored $\red$), our fixing policy specifies that we color $u$ $\green$.\footnote{If the star is a subgraph, try to color it $\green$ first; if $\green$ is unavailable, color it using some other color.} Let $A = \{\ca_i\}$ be the set of colorings such that $\ca_i(v)=\blue$, $\ca_i(w)=\red$,   and  the vertices $u \notin \{v,w\}$ satisfy either $\ca_i(u) = \red$ or $\ca_i(u) = \green$. Let $\cb$ be such that $\cb(v)= \red$ and $\cb(u)=\green$ for all $u \neq v$.  Now consider $t$, the last transition in the path from any $\ca_i$ to $\cb$. Note that this transition is the same for all of these paths: $v$ is $\red$ and all of the leaves are $\green$, except for $w$, which is $\red$. The size of $|A|$ is $2^{n-2}$, hence at least $2^{n-2}$ paths, all of whose destination is $\cb$, use $t$; the canonical paths argument described above falls short.

To get around these difficulties, we use a multicommodity flow argument~\cite{Sinclair92,Diaconis91}. Instead of  specifying a single path for each pair of states, we describe a flow between them in the underlying graph. Whenever we fix a vertex's color, we split the flow evenly among all available options. This is similar to the argument of Morris and Sinclair~\cite{MorrisS04}, in which flow is also split up among different paths; it helps route the  flow more evenly, thereby avoiding the case where one edge is heavily congested while other ``available'' edges are not.
In contrast to~\cite{MorrisS04}, we only split the flow  in the fixing stages; in fact, whenever a vertex $v_i$ is colored  $\cb(v_i)$, this consolidates the flow!  Our technique can be thought of as a hybrid argument between the canonical paths proof technique of~\cite{JerrumS89} and the multicommodity flow argument of~\cite{MorrisS04}.

Unfortunately, splitting  the flow is still not enough to guarantee a sufficiently low congestion: if many vertices are being ``fixed'' at the same time, there are not necessarily enough available colors  to spread the flow thinly enough, as the number of available colors for each vertex is possibly only  $\eps\Delta$, while it could have potentially been colored with any one of the $(1+\eps)\Delta$ colors. When the pathwidth of the graph is bounded by $O(\log{n})$, there is an ordering that guarantees good \emph{vertex separation}~\cite{Kinner}. We define this formally in Section~\ref{sec:treewidth}, but in essence it implies that there is some order on the vertices, such that if our canonical paths obey this order, not too many vertices will need to be fixed at any time. We note that finding such an order is $NP$-hard~\cite{NP}, however we only need that such an order exists for the canonical paths argument. 

Finally, recall that the state space of the Single-Flaw dynamics includes flawed colorings. We show that there are not too many singly-flawed colorings relative to proper colorings, hence executing the chain polynomially many times will guarantee that we output a proper coloring w.h.p.

\subsection{Related work}
\label{sec:related}

Most of the work on sampling colorings has focused on the Glauber dynamics (e.g.,~\cite{Jerrum95,Dyer10,DyerFHV13,HayesVV15,Goldberg10,Tetali10,Lu17,Mossel10,Lucier11,Molloy04}). Other Markov chains have been analyzed, notably the Flip dynamics of Vigoda~\cite{Vigoda99}, which is closely related to the chain proposed by Wang,  Swendsen,  and Koteck\'y~\cite{Pot}: when a vertex $v$ is required to change color from $c$ to $c'$, the colors of the entire neighborhood that is colored with $c$ and $c'$ are flipped (the chain of Vigoda only performs some flips with some probability). It is also possible to sample colorings using approaches that do not use MCMC methods; for example, Efthymiou~\cite{Ef2} proposed a combinatorial method for sampling colors that does not use a Markov chain, and used it to show that it is possible to sample colorings  on $G(n,d/n)$ using $k>(1+\eps)d$ colors. A caveat is that the run time is only polynomial w.p. $1-2n^{-2/3}$. 
The main difference between Single-Flaw dynamics and other work on sampling colors is that we allow flawed vertices (specifically, one at a time). There are other Markov chains that also consist of ``flawed'' states that are not part of the space we wish to sample from. An example is the Markov chain for sampling perfect matchings in bipartite graphs, proposed by Broder~\cite{Broder} and analyzed by Jerrum and Sinclair~\cite{JerrumS89} and Jerrum, Sinclair and Vigoda~\cite{JerrumSV04}: the chain consists of perfect matchings (of size $n$), and imperfect matchings of size $n-1$.  An interesting distinction is that Broder's chain needs the imperfect matchings to transition between perfect matchings (otherwise, it is unclear how to transition). We do not need the imperfect states to show convergence: the Glauber dynamics are  known to converge to the uniform distribution for $k \geq \Delta+2$; we only use the imperfect colorings to bound the mixing time.

The method of bounding the conductance or congestion of the transition matrix of a Markov chain has been used to great success in sampling and counting of various problems e.g.,~\cite{JerrumSV04,GuoJ17,MorrisS04}. 
To our knowledge, the only place that these types of arguments have been successfully applied to sampling colorings is in bounding the mixing time of the Glauber dynamics on trees with bounded degree~\cite{Lucier11}. The arguments of~\cite{Lucier11} rely heavily on both the tree structure and the upper bound on the degree and it seems very difficult, if at all possible, to extend their techniques to more general settings such as ours.

Sampling colorings corresponds to sampling
configurations of the zero temperature $k$-state anti-ferromagnetic Potts model~\cite{Potts}. One can draw an analogy between our technique and temperature-tuned walks that also include higher energy levels (see e.g.,~\cite{Liu}),  though
instead of  walking at a fixed temperature, which would allow some Poisson-like distribution of the number of flaws, we allow exactly one flaw, and correct it before allowing the next flaw.

The terms \emph{treewidth} and \emph{pathwidth} were introduced by Robertson and Seymour~\cite{Tree,Path}; the concept of treewidth was discovered independently several times, and was originally introduced under a different name by Bertel\`{e} and Brioschi~\cite{Bertele}.
There is a large literature of graphs studying graphs bounded treewidth, e.g.,~\cite{Andrzejak,Gupta,Loksh,Fomin,Arn,Bruno}. Of particular interest is a work by Chekuri, Khanna and Shepherd~\cite{Chekuri} that also considers multicommodity flows on graphs of bounded treewidth.  Their techniques and results are incomparable to ours; they study the flow on graphs, while we use a flow to bound the congestion on the underlying graph of the Single-Flaw dynamics.

\section{Preliminaries}
\label{sec:prelims}
We denote the set $\{1,2,\ldots, m\}$ by $[m]$. 
Let $G=(V,E)$ be a graph, and denote $|V|=n$.
We assume that the vertices of $G$ are uniquely identified by $\{1, 2,\ldots, n\}$. For any (not necessarily simple) path $p$ in $G$, let $|p|$ denote the length of $p$ (i.e., the number of edges in $p$, where if an edge appears $k$ times in $p$, it is counted $k$ times). 

\subsection{Colorings}

For any $k$-coloring of $G$, $\sigma:V\rightarrow[k]$, let $\mono=\{(u,v):\sigma(u)=\sigma(v)\}$ denote the set of monochromatic edges. $\sigma$ is a proper coloring if $\mono=\emptyset$. $\sigma$ is a \emph{singly-flawed coloring} if $\mono\neq \emptyset$  and there is a vertex that is common to all edges in $\mono$, i.e., $\exists v: \forall e\in \mono, v \in e$. We say that such a vertex $v$ is  a \emph{flawed vertex} of $\sigma$. 
Note that a singly-flawed coloring has exactly two flawed vertices if $|\mono|=1$ and one flawed vertex otherwise. We denote the set of proper $k$-colorings of $G$ colors by $\C_p(G,k)$ and the set of all singly-flawed colorings by $\C_{sf}(G,k)$. We drop  $G$ and $k$ when they are clear from context. Let $\sigma$ be a coloring. If, after recoloring some $v\in V$ with a color $c$, there is no monochromatic edge $(u,v)$,  we say that  $c$ is \emph{available} to $v$ in $\sigma$. Note that a color's availability does not depend on whether $\sigma$ or the coloring obtained by recoloring $v$ with $c$ is proper, singly-flawed or otherwise.

We first show two results that will be useful later on, regarding proper and singly-flawed colorings: (1) the ratio of  singly-flawed colorings to proper colorings is ``not too large'' (Corollary~\ref{cor:same}), and (2)  there is a mapping from singly-flawed colorings to proper colorings, such that ``not too many'' singly-flawed colorings are mapped to any proper coloring (Corollary~\ref{cor:sur}). Both results are corollaries of the following simple lemma.

\begin{lemma}\label{lem:same}
	For any $G=(V,E)$ such that $|V|=n$ and $k \geq \Delta+2$, there exists a surjective function 
	$$g:\C_p(G,k) \times [k]\times[n] \rightarrow \C_{sf}(G,k).$$
\end{lemma}
\begin{proof}
	For every coloring $\sigma \in \C_p(G,k)$, every vertex $v \in V$ and every color $c \in [k]$, let \begin{align*}
	\sigma'_{c,v} = 
	\begin{cases}
	\sigma(u) &\text{if }u \neq v\\
	c& \text{if }u = v
	\end{cases}
	\end{align*}
	If $\sigma'_{c,v}  \in \C_{sf}(G,k)$, let $g(\sigma,c,v) = \sigma'_{c,v}$, otherwise let $g(\sigma,c,v)$ be some arbitrary coloring in $\C_{sf}(G,k)$.  
	It is easy to see that every $\sigma' \in \C_{sf}(G,k)$ is in the range of $g$: the reverse operation of changing the color of a flawed vertex $v$ in $\sigma'$ to some available color  gives a proper coloring.
\end{proof}
The two corollaries that we require are the following.
\begin{corollary}\label{cor:same}
	For any $G=(V,E)$ such that $|V|=n$ and $k \geq \Delta+2$,
	$$|\C_{sf}(G,k)| \leq kn|\C_p(G,k)|.$$
\end{corollary}
\begin{proof}Immediate from the surjectivity of the function $g$ in Lemma~\ref{lem:same}.
\end{proof}
\begin{corollary}\label{cor:sur}
	For any $G=(V,E)$ such that $|V|=n$ and $k \geq \Delta+2$, there exists a function 
	$$g':\C_{sf}(G,k) \rightarrow \C_{p}(G,k),$$
	for which each element in the co-domain has at most $kn$ pre-images in the domain.
\end{corollary}
\begin{proof}
	For every $\sigma' \in \C_{sf}$,  arbitrarily select  one pre-image $(\sigma,v,c)$ w.r.t. $g$, and set $\sigma$ as the image for $\sigma'$ under $g'$.
\end{proof}

\subsection{Markov chains and rapid mixing}
\label{sec:markovprelims}

In this section we review some of the results on the mixing time of Markov chains that we will require. The reader is referred to~\cite{LevinPW06} for an excellent introduction to Markov chains and modern techniques on bounding their mixing time. 

Consider  a discrete-time Markov chain $\mc$ with finite state space $\Omega$ and symmetric transition probability matrix $P$ (i.e., $P(\sigma, \sigma') = P(\sigma', \sigma)$ for all $\sigma, \sigma' \in \Omega$). The chain is said to be \emph{irreducible} if for every pair of states $\sigma, \sigma' \in \Omega$, there exists some $t$ such that $P^t(\sigma, \sigma')>0$; in other words, it is possible to get from any state to any state using a finite number of transitions. It is \emph{aperiodic} if for any $\sigma,  \in \Omega$,  $\text{gcd}\{t:P^t(\sigma, \sigma)>0\}=1$. It is \emph{lazy} if for all $\sigma \in \Omega, P(\sigma, \sigma)>1/2$. A fundamental theorem of stochastic processes states that an irreducible and aperiodic Markov chain converges to a unique \emph{stationary distribution} $\pi$ over $\Omega$, i.e., $\lim_{t\to\infty}P^t(\sigma, \sigma')= \pi(\sigma')$ for all $\sigma, \sigma' \in \Omega$.  
If in addition $P$ is symmetric,  then $\pi$ is uniform over  $\Omega$ (e.g.,~\cite{Aldous1983}).

Our goal is to describe  a \emph{fully-polynomial almost uniform sampler for proper colorings}; namely, 
a randomized algorithm that, given as inputs a graph $G=(V,E)$ and  a bias
parameter $\delta$, outputs a random proper coloring of $G$ from a distribution $D$ that satisfies
$d_{TV}(D, U) \leq \delta$, where $U$ is the uniform distribution on the proper colorings of $G$ and $d_{TV}$ is the total variation distance, defined as follows:\footnote{Alternatively, we can define it as $d_{TV}(\mu, \nu)=\frac{1}{2}\sum_{\sigma \in \Omega} |\mu(\sigma)-\nu(\sigma)|$. It is easy to verify that the two definitions are equivalent.} For any two distributions $\mu, \nu$ on $\Omega$,

\begin{equation}\label{def:dtv}
d_{TV}(\mu, \nu) = \max_{S \subseteq \Omega}|\mu(S)-\nu(S)|.
\end{equation}

We are interested in the rate at which a  Markov chain converges to its stationary distribution $\pi$.  We define the \emph{mixing time} from a state $\sigma$ to be 

\begin{equation}
\tau_\sigma(\delta)=  \min\{\bar{t}: d_{TV}( P^t(\sigma, \cdot), \pi)\leq\delta\ \text{ for all } t \geq \bar{t}\},
\end{equation}

We further define the mixing time of the Markov chain to be  $\tau(\delta)=\displaystyle\max_{\sigma}\tau_\sigma(\delta)$. We say that a Markov chain is \emph{rapidly mixing} if $\tau(1/2e)$ is polynomial in $n$. The constant $1/2e$ is arbitrary, as a bound on $\tau(1/2e)$ implies a bound on $\tau(\delta)$ for any $\delta>0$ (e.g.,~\cite{Aldous1983}):
$$\tau(\delta)\leq (1-\log{\delta})\cdot \tau(1/2e).$$

In order to bound the mixing time, we describe a multicommodity flow on the underlying graph $H=(\Omega, F)$ of the Markov chain, where $F = \{ (\sigma,\sigma'), P(\sigma, \sigma')>0 \}$ is the set of all transitions that have positive probability.

We denote by 
\begin{equation}\label{eq:q} q(\sigma,\sigma') = \pi(\sigma)P(\sigma, \sigma'),\end{equation} 
the \emph{ergodic flow} through the edge $(\sigsig)$ of $H$ (an intuitive way to think about $q(\sigma,\sigma')$ is the probability of traversing edge $(\sigma, \sigma')$ at stationarity.) 

For all ordered pairs $(\cc)\in \Omega^2$, let $\paths_{\cc}$ denote a set of (not necessarily simple) directed paths from $\ca$ to $\cb$ in $H$. 
A \emph{flow} is a function $f:\paths \rightarrow \R^+\cup \{0\}$ where $\paths = \bigcup_{\cc}\paths_{\cc}$ 
that satisfies 
\begin{equation}\label{eq:flow}
f_{\cc} \equiv \sum_{p \in \paths_{\cc}} f(p) = \pi(\ca)\pi(\cb),
\end{equation} for   every $\cc \in \Omega$. 

We define the congestion on an edge $(\sigsig)$ with respect to a flow $f$ by 
\begin{equation}\label{eq:congestion}
\rho_f(\sigsig) = \frac{1}{q(\sigsig)} \sum_{\cc \in \Omega} \sum_{p:(\sigsig) \in p \in \paths_{\cc}} f(p)|p|,	\end{equation}

and the congestion of $f$ by 
$$\rho_f = \max_{(\sigsig) \in F}\rho_f(\sigsig).$$

We use the following theorem, due to Sinclair~\cite{Sinclair92} and  Diaconis and Stroock~\cite{Diaconis91}, that relates the mixing time to the congestion of a flow. Note that it holds for any flow; in order to bound the mixing time, we need to find \emph{some} flow that has low congestion.
\begin{theorem}\cite{Sinclair92}\label{thm:sinclair}
	For any irreducible, aperiodic, lazy and symmetric Markov chain $\mc$ with transition matrix $P$ on state space $\Omega$,  any flow $f$ on the underlying graph of $\mc$, and any state $\sigma_0 \in \Omega$,
	\begin{equation*}
	\tau_{\sigma_0}(\delta) \leq \rho_f  \left( \ln \pi(\sigma_0)^{-1} + \ln \delta^{-1}\right). 
	\end{equation*}
\end{theorem}

\subsection{Treewidth, pathwidth and vertex separation}
\label{sec:treewidth}

A \emph{tree decomposition} of a graph $G=(V,E)$ is a tree $T$ with $m$ nodes, where each node of $T$ represents a subset of $V$: $X_i \subseteq V, i \in [m]$ such that the following hold:
\begin{enumerate}
	\item $\bigcup_{i=1}^m X_i = V$,
	\item For every $(u,v) \in E$,  $u, v \in X_i$ for some $i \in [m]$.
	\item For all $i,j,k \in [m]$, if $X_j$ is on the (unique) path between $X_i$ and $X_k$, then $X_i \cap X_k \subseteq X_j$.
\end{enumerate}
The first requirement guarantees that every vertex of $G$ is in at least one node of $T$, the second that every two neighboring vertices in $G$ share at least one node in $T$, and the third that if some vertex $v \in V$ is in both $X_i$ and $X_k$, it is in every node of the path between $X_i$ and $X_k$ in $T$.
The width of $T$ is defined as $\max_{i \in [m]}\{|X_i|-1\}$.
The treewidth of $G$, denoted $\tw(G)$, is the minimal $\omega$ such that there exists some tree-decomposition of $G$ with width $\omega$.
The \emph{pathwidth} of $G$ (denoted $\pw(G)$),
is defined analogously to treewidth, with $T$ constrained to be a path.

A linear ordering  of a graph $G=(V,E)$ is a bijective mapping of  vertices to integers; $\lgg:V \rightarrow \{1, 2, \ldots, n\}$. Given a graph $G$, a linear ordering $\lgg$, and an integer $j \in \{1,\ldots, n\}$, let $A_j$ be the set of vertices mapped to the integers $1, \ldots, j$ by $\lgg$; i.e., $A_j = \{v:\lgg(v) \leq j\}$. Let $B_j$ denote the set of vertices that are mapped to integers greater than $j$ by $\lgg$: $B_j = V \setminus A_j$. A \emph{minimal vertex separator} for an index $j \in \{1,\ldots n\}$ (denoted $\mvs(G,\lgg,j)$) is a minimal set of vertices $S_j \subset V$ such that the following hold:\footnote{Traditionally, a vertex separator for  $j$ is defined such that the separating subset appears before $j$ in the ordering~\cite{Ellis94}. This is identical to our definition with the order inverted.}
\begin{enumerate}
	\item For all $u \in S_j$, $\lgg(u)>j$. \item $A_j$ and $B_j \setminus S_j$ are disconnected; that is, there is no edge $(u,v) \in E$ such that $u \in A_j, v \in B_j \setminus S_j$. 
\end{enumerate}
The \emph{vertex separation number} of a graph $G$ and a linear order $\lgg$, denoted $\vsn(G,\lgg)$ is the size of the largest minimal vertex separator. That is, $$\vsn(G,\lgg) = \max_{j \in \{1,\ldots,n\}} \{\mvs(G,\lgg,j)\}.$$ 
The \emph{vertex separator number} for a graph $G$, denoted $\vsn(G)$ is the minimal vertex separation number $\vsn(G,\lgg)$ over all possible linear orderings $\lgg$ of $G$. 
$$\vsn(G) = \min_{\lgg}\vsn(G,\lgg).$$
We call an order $\lgg$ for which $\vsn(G)=\vsn(G,\lgg)$ a \emph{minimal order} for $G$.

We require the following two theorems, relating the treewidth, pathwidth and vertex separation number of a graph.
\begin{theorem} \cite{Kinner} \label{Kinner}For any graph $G$, 
	$\vsn(G) = \pw(G)$.
\end{theorem}
\begin{theorem}\cite{Korach} \label{Korach} For any graph $G$, 
	$\pw(G) = O(\tw(G) \log{n})$.
\end{theorem}

\section{Single-Flaw dynamics}

Let $G=(V,E)$ be a  graph with maximal degree $\Delta$ and  $\eps>0$ such that $(1+\eps)\Delta \geq \Delta+2$. 
The state space $\Omega$ of  Markov chain $\mc{(G,\eps)}$ (or simply $\mc$) is the set of all proper and singly-flawed $k$-colorings of $G$, for $k=\lceil(1+\eps)\Delta\rceil$: $\Omega = \C_p(G,k) \cup \C_{sf}(G,k)$. For simplicity, we henceforth assume that $\eps\Delta$, $(1+\eps)\Delta$ and $(1+\eps)\eps^{-1}$ are integers. It is easy to generalize the  results to real values thereof.
For $\sigma \in \Omega$, the transitions $\sigma \rightarrow \sigma'$ of $\mc$ are the following

\begin{itemize}
	\item Let $\sigma'=\sigma$.
	\item With probability $1/2$, do nothing (laziness).
	\item Otherwise, choose a vertex $v$ and color $c$ uniformly at random from $V$ and $[k]$ respectively.   Tentatively, set $	\sigma'(v)=	c$
	\item If $\sigma' \notin \Omega$, set $\sigma'(v) = \sigma(v)$.
\end{itemize}

It is easy to verify that the chain is irreducible, aperiodic, lazy  and symmetric;  hence the conditions of Theorem~\ref{thm:sinclair} hold, and it remains to describe a flow with low congestion of the underlying graph of $\mc$.
Our main result describes such a flow.
\begin{lemma}\label{lemma:main}
	Let $\eps>0$ and $G$ be a graph with maximal degree $\Delta$. 
	Then there exists a flow $f$ on the underlying graph of the Markov chain $\mc{(G, \eps)}$ such that
	$$\rho_f \leq   8\pw(G)(1+\eps)^3\Delta^3 n^5  \left( (1+\eps)\eps^{-1}\right) ^{2\pw(G)}.$$
\end{lemma}

There are at most $k^n$ possible  colorings of $G$. 
Because the stationary distribution is uniform, for all $\sigma \in \Omega$, $\pi(\sigma) \geq \frac{1}{k^n}$, hence  $\ln \pi(\sigma)^{-1} = O(n \log{n})$.

Theorem~\ref{thm:sinclair} and Lemma~\ref{lemma:main} together imply Theorem~\ref{thm:main}, which is  formally restated as follows.
\begin{maintheorem}
	Let $\eps>0$ and $G$ be a graph with  maximal degree $\Delta$. The mixing time of  $\mc{(G, \eps)}$ satisfies
	$$\tau(\delta) =   O\left( \pw(G)(1+\eps)^3\Delta^3 n^5   \left( (1+\eps)\eps^{-1}\right) ^{2\pw(G)} \left( n\log{n} +\ln \delta^{-1}\right)\right) .$$
	In particular, if $\pw(G) = O(\log{n})$, the chain mixes in polynomial time.
\end{maintheorem}

In order to prove Lemma~\ref{lemma:main}, we design a flow for $\mc$. To do so, we first describe the set of canonical paths\footnote{We note that the term ``canonical paths'' is traditionally used for describing  a single path between every pair of states. In our case, it means a set of paths for each state.} that we will route the flow through.

\subsection{Canonical paths}

Let $L$ be some minimal order of $G$. For simplicity and w.l.o.g. we assume that $L$ is the identity order, i.e., $L(i)=i$ for all $i\in [n]$. 
We remark that we do not need to explicitly find $L$; we only require its existence for the canonical paths argument. Let $\lambda = \frac{\vsn(G)}{\log{n}}$. If $\pw(G)=O(\log{n})$, as is assumed here, $\lambda$ is a constant. 

For each $j \in [n]$, let $S_j=\mvs(G,L,j)$ be a minimal vertex separator.
We denote the set of canonical paths from $\ca$ to $\cb$ by $\gamma_{\cc}$. \emph{ For the rest of this subsection and the next, we assume that $\ca$ and $\cb$ are both proper colorings}; we will extend the sets of  paths to include ones that start and/or end at singly-flawed colorings in Section~\ref{sec:genflows}.  We divide each path into $n$ phases, where phase $j$ consists of $|S_j|+1$ steps, for a total of $\length=n+\sum_{j=1}^n |S_j| \leq (\lambda+1) (n\log{n})$ steps. Each step is a recoloring of some vertex; it is possible that a vertex is ``recolored'' with the same color. In that case, the state (coloring) does not change, but we still count this redundant recoloring as a step, as it guarantees that all  paths are of the same length; this will help to make the analysis more concise. A state that appears at the start of the $\ell^{th}$ step of the $j^{th}$ phase of a canonical path in $\gamma_{\cc}$ is said to be at \emph{distance} $(j,\ell)$ from $\ca$; alternatively, we say that it happens at \emph{time} $(j,\ell)$. 
We denote the states of $\gamma_{\cc}$ that are at distance $(j,\ell)$ from $\alpha$ by $\states_{\cc}(j,\ell)$. 
The set of canonical paths $\gamma_{\cc}$ can be thought of as a layered graph, where all states of $\states_{\cc}(j,\ell)$ are placed in the same layer. 

\subsubsection{A phase of the canonical paths}
We  describe a single phase of $\gamma_{\cc}$. For any $j \in [n]$, all states in $\states_{\cc}(j,1)$ are proper colorings. Note that $\states_{\cc}(1,1) = \{\alpha \}$. 
For {\bf the  first step} of the $j^{th}$ phase, for every $\sigma_i \in \states_{\cc}(j,1)$, set 

\begin{align*}
\sigma'_i(v) = 
\begin{cases}
\sigma_i(v) &\text{if }v \neq j\\
\beta(j)& \text{if }v = j
\end{cases}
.
\end{align*}

We therefore have that  $\states_{\cc}(j,2) = \bigcup_i \sigma'_i$. It is clear that there is only one way to route the flow entering $\sigma_i$: it is all routed to $\sigma'_i$ on $(\sigma_i, \sigma'_i)$. It is possible that flow becomes consolidated in this step: the flow from all states $\sigma_i \in \states_{\cc}(j,1)$ that differ only in the $j^{th}$ coordinate is routed to the same $\sigma'_i$.
Note that every $\sigma' \in \states_{\cc}(j,2)$ is either a  proper or a  singly-flawed coloring.  

{\bf The $\ell^{th}$ step} in the $j^{th}$ phase, $\ell \in \{2, 3, \ldots |S_j|+1\}$  is  a \emph{splitting step}, and is the following:
Let  $u_\ell$ be the (lexicographically) $(\ell-1)^{th}$ vertex  of $S_j$.  
For every $\sigma_i \in \states_{\cc}(j,\ell)$, let $C_i(u_\ell)$ be the set of  colors available to $u_\ell$ under $\sigma_i$. For each $\sigma_i \in \states_{\cc}(j,\ell)$ and color $c \in C_i(u_\ell)$, let

\begin{align*}
\sigma'_{i,c}(v) = 
\begin{cases}
\sigma_{i}(v) &\text{if }v \neq u_\ell\\
c& \text{if }v = u_\ell 
\end{cases}
.
\end{align*}

We have that $\states_{\cc}(j,\ell+1) = \bigcup_i \sigma'_{i,c}$. From each state $\sigma_i \in \states_{\cc}(j,\ell)$, the flow is split  evenly among the transitions (i.e., a $1/|C_i(u_\ell)|$ fraction of the flow entering $\sigma_{i}$ is routed on each $(\sigma_i,\sigma'_{i,c})$). Note that all states in  $\states_{\cc}(j,|S_j|+2)$ are  proper colorings, as any edge that may have been monochromatic in  any $\sigma \in \states_{\cc}(j,2)$ will have been recolored.  For all $1\leq j<n$, set $\states_{\cc}(j+1,1) = \states_{\cc}(j,|S_j|+2)$.
Note that $S_n=\emptyset$ and $\states_{\cc}(n,2)=\{\beta\}$.

Because for all $j$, $|S_j| = O(\log{n})$, given $\alpha, \beta, j$ and $\ell$,  the color of most vertices in $\states_{\cc}(j,\ell)$ is uniquely determined. In particular, for $\ell>1$, denote $A_j = \{v: v\leq j\}$ and $B_j = V \setminus (A_j \cup S_j)$.\footnote{For $\ell=1$, we consider $(j-1,|S_{j-1}|+2)$ instead of $(j,1)$, unless $j=1$, in which case the vertices are all colored by $\ca$.} It must hold that for any $\sigma \in \states_{\cc}(j,\ell)$, $v_a \in A_j$, and $v_b \in B_j$,  $\sigma(v_a)=\beta(v_a)$ and $\sigma(v_b) = \alpha(v_b)$. This is because all the vertices in $A_j$ have been recolored to $\beta$ and will not be recolored again, while the vertices in $B_j$ have no neighbors in $A_j$, hence they have not been recolored yet.

We denote by $\Q(\alpha,\beta,j,\ell)$ ($\Q$ stands for ``quantum set'') the number of vertices whose color is not uniquely defined by $\alpha, \beta, j, \ell$. We note that $\Q(\alpha,\beta,j,\ell) \subseteq S_j$, but that equality does not necessarily hold:  assume that $S_{j-1} \subset S_{j}$, and let $u \neq j$ be some vertex in $S_j \setminus S_{j-1}$. Then $u$'s color is still $\ca(u)$ at time $(j,2)$, as it has not yet been recolored, even though $u \in S_j$.

We note that $\Q(\alpha,\beta,j,\ell)$ does not in fact depend on $\ca$ or $\cb$. In fact,
\begin{obs}\label{obs:1}$\Q(\alpha,\beta,j,\ell)$
	is uniquely determined by either
	\begin{enumerate}
		\item $j$ and $\ell$, or
		\item $j$ and $(\sigsig)$.
	\end{enumerate}
\end{obs}
\begin{proof}
	For any $\cc$, the same vertex is recolored at  $(j,\ell)$: at $\ell=1$, vertex $j$ is recolored; in all other instances, $u_\ell$ is recolored by at least $\eps\Delta$ different colors, regardless of $\cc$. Further, note that once a vertex $u$ is in such a ``quantum state'', it will remain in quantum state until it is colored $\beta(u)$ at distance $(u,1)$. Therefore, although we cannot recover the exact transitions used without knowledge of $\cc$, the set of vertices whose color is unknown at any given time is fixed. For the second observation, notice that $(\sigsig)$ recolors some specific vertex $u_\ell$ (even if it is an idle recoloring), hence $\ell$ can be inferred. 
\end{proof}
Due to Observation~\ref{obs:1}, we sometimes refer to $\Q(\alpha,\beta,j,\ell)$ by  $\Q(j,(\sigsig))$ or $\Q(j,\ell)$, depending on the context.

\subsection{Bounding the flow}
Let $\ca$ and $\cb$ be proper colorings, $t=(\sigsig) \in F$ be some transition, and $j \in [n]$ be an integer. Denote the flow routed through $t$ from $\ca$ to $\cb$ in phase $j$ by $f_{j,t,\ca,\cb}$.  Note that we are only considering the flow routed through $t$ {\bf during phase $j$}; it is possible that the canonical path passes through $t$ in several phases, possibly carrying a different flow each time. Intuitively, it seems natural that   after the flow was split evenly several times, ``not too much'' flow is routed through any state, as it has a specific combination of colors of the vertices of $\Q(j,t)$, and many such combinations are possible. It is not straightforward to show this, however, as   the colors of different vertices in each state are not independent.  The difficulty is compounded by the fact that flow is consolidated at the first step of every phase. Nevertheless, we can prove the following claim by rearranging the vertices of $\Q(j,t)$ and using an inductive reasoning on this new order.

\begin{claim}\label{claim:fuckme}
	The flow routed from $\ca$ to $\cb$  through any  $t=(\sigsig) \in F$ in any phase $j \in [n]$  is at most $$f_{j,t,\ca,\cb}\leq\frac{\pi(\ca)\pi(\cb)}{\left( \eps\Delta\right)^{|\Q(j,t)|} }.$$
\end{claim}

\begin{proof} From Observation~\ref{obs:1}, $(j,t)$ uniquely defines $(j,\ell)$. If no flow is routed from $\ca$ to $\cb$ through $t$ in phase $j$, the claim is trivially satisfied. Otherwise, we show that 

	\begin{equation*}\label{eq:ll}f_{j,\ell,\ca,\cb}\leq\frac{\pi(\ca)\pi(\cb)}{\left( \eps\Delta\right)^{|\Q(j,\ell)|} },
	\end{equation*}
	where $f_{j,\ell,\ca,\cb}$ is the maximal flow from $\ca$ to $\cb$ through any $t$ at distance $(j,\ell)$ from $\ca$.	Order the vertices of $\Q(j,\ell)$ in reverse order of the time since their last color change (possibly a null color change). That is, the vertex whose color changed most recently is last in the order. Let $M=|\Q(j,\ell)|$ and relabel the vertices of $\Q(j,\ell)$ by  $1, \ldots, M$ according to their place in this order. Similarly, relabel $\states_{\cc}(j',\ell')$, by $\states_{1}, \ldots, \states_{M}$, where $\states_{i}$ is the set of states at the time just \emph{after} vertex $i$ last changed its color. In other words, $\states_M=\states_{\cc}(j,\ell+1)$,$\states_{M-1}=\states_{\cc}(j,\ell)$, and so on. It is possible that for some $m$, $\states_m$ corresponds to states in the previous phase, i.e., $\states_m =\states_{\cc}(j-1,\ell')$.  
	Note that we drop the $\cc$ from the notation for clarity, but we are still only considering  the flow from $\ca$ to $\cb$.

	We now show by that   for any set of $m \leq M$ colors $c_1, \ldots, c_m$,    at most $\frac{\pi(\ca)\pi(\cb)}{(\eps\Delta)^m}$ flow is routed into $\{\sigma \in \states_{m}:\sigma(i)=c_i, i \in [m]\}$. In other words, fix the colors $c_1, \ldots, c_m$. We want to bound the flow that passes through (into) the states of $\states_m$, where vertices $1,\ldots,m$ are colored with $c_1, \ldots, c_m$ respectively. We do this by induction on $m$. 
	
	{\bf The base case:} 
	
	The total flow from $\ca$ to $\cb$ through $\states_i$, for any $i$, is exactly $\pi(\ca)\pi(\cb)$. For any $c_1 \in [k]$, at most $\frac{\pi(\ca)\pi(\cb)}{\eps\Delta}$ flow is routed through $\{\sigma \in \states_{1}:\sigma(1)=c_1\}$. This is because the last time vertex $1$ changed color, at most $1/\eps\Delta$ of all the flow   was routed to states where $v$'s color is $c_1$. 
	
	{\bf The inductive step:}
	
	From the inductive hypothesis, at most  $\frac{\pi(\ca)\pi(\cb)}{(\eps\Delta)^{m-1}}$ flow is routed through $\{\sigma \in \states_{m-1}:\sigma(i)=c_i, i \in [m-1]\}$. From the construction of the canonical paths, for each of these states, at most  $1/\eps\Delta$ of the flow entering it flows to a state where vertex $m$ is colored $c_m$.
	
\end{proof}
We want to bound the total flow through a transition in any single phase. The following set of recoloring functions $\chi$ is useful. Let  $C$ be a set of (available) colors. $\chi_C$ is a function, parameterized by $C$, that takes as an input a  color $c \in [k]$. Its output is a color from $C$, such that each color in $C$ has the same number of pre-images, up to one.   We do not explicitly define $\chi$, only note that such a set of functions exists. For example, if $k=13$, $C=\{1,2,3,4,5\}$,  $\chi_C$ could allocate $(c \mod 5) +1$ to every $c \in [k]$, giving each color in $C$ either two or three pre-images. We make the following observation (recall we assume $\eps\Delta$, $(1+\eps)\Delta$ and $(1+\eps)\eps^{-1}$ are integers). 
\begin{obs}\label{obs:ell}
	For $\chi_C$ as defined above, if $k= (1+\eps)\Delta$ and $|C| \geq \eps\Delta$, each color in $C$ has at most $(1+\eps)\eps^{-1}$ pre-images in $[k]$.
\end{obs}

Armed with Claim~\ref{claim:fuckme} and the functions $\chi$, we are now ready to bound the total flow through a transition in any single phase. 
\begin{lemma}\label{lem:fl}
	The flow $f_{j,t}$ routed through any $t=(\sigsig) \in F$ in any phase $j \in [n]$ satisfies $$f_{j,t} \leq \pi(\cdot)^2|\C_p| \cdot  \left( (1+\eps)\eps^{-1}\right) ^{2|S_j|},$$
	where $\pi(\cdot)$ is the probability of any state at stationarity.
\end{lemma}
\begin{proof}
	For each $(j,t)$, where $j \in [n]$ and $t=(\sigsig) \in F$, 
	denote by
	$\cp_{j,t}$ the set of pairs of states $\cc \in \C_p^2$ whose  paths pass through $t$ in phase $j$. 
	We describe a function $\mu_{j,t}$ whose domain is $\cp_{j,t}$. 
	We view the  co-domain of $\mu_{j,t}$  as the Cartesian product of $3$ sets $X$, $Y$ and $Z$:
	$\mu_{j,t}:\cp_{j,t} \rightarrow X \times Y \times Z$; the output of $\mu_{j,t}(\cdot)$ is a triple $(x,y,z)$. The function will be injective, therefore the size of the co-domain  of $\mu_{j,t}$ will serve as an upper bound to $|\cp_{j,t}|$.
	\medskip

	The sets $X,Y,Z$ are the following.

	\begin{itemize}
		\item $X$ is the set of all proper colorings. Assume that the input to $\mu_{j,t}$ is some pair $(\ca,\cb)$. 
		In the coloring specified by $x$, all vertices in $A_j$\footnote{As before, $A_j = \{v: v\leq j\}$, except for $\ell=0$, for which $A_j = \{v:v<j\}$.}  are  colored by $\ca$. All vertices in $B_j$ are colored by $\cb$. Before specifying the coloring of $S_j$ under  $x$, note that already, together with $j$ and $t$, this allows us to deduce $\alpha$ completely on all vertices in $V \setminus \Q(j,t)$ and $\beta$ on all vertices $V \setminus S_j$.
		To determine the colors of $S_j$ in $x$, we color them one at a time, using $\chi$. 
		This information, while not characterizing $\beta(v)$ completely for $v \in S_j$, allows allows us to restrict the possible value of $\beta(v)$ to a set of size at most $(1+\eps)\eps^{-1}$ possible values. 
		\item $Y$ is $\left[ (1+\eps)\eps^{-1}\right]^{|S_j|}$, allowing us to pinpoint $\beta(v)$ for every $v \in S_j$.
		\item Finally, $Z$ is simply all possible colorings of the vertices of $\Q(j,t)$ under $\alpha$. 
	\end{itemize}
	Clearly $x,y,z,j$ and $t$ allow us to recover $\ca$ and $\cb$.
	The size of the co-domain is at most $$|\C_p| \cdot  \left( (1+\eps)\eps^{-1}\right) ^{|S_j|} \cdot k^{|\Q(j,t)|}.$$
	Combining with Claim~\ref{claim:fuckme} we get that the total flow through any transition $t$ at phase $j$ is at most
	
	\begin{align*}f_{j,t} & \leq  |\C_p| \cdot  \left( (1+\eps)\eps^{-1}\right) ^{|S_j|} \cdot k^{|\Q(j,t)|}\cdot\frac{\pi(\cdot)\pi(\cdot)}{\left( \eps\Delta\right)^{|\Q(j,t)|} }\\
	&=  \pi(\cdot)^2|\C_p| \left( (1+\eps)\eps^{-1}\right) ^{|S_j|} \cdot  \frac{  ((1+\eps)\Delta)^{|\Q(j,t)|}}{\left( \eps\Delta\right)^{|\Q(j,t)|}}\\
	&\leq \pi(\cdot)^2|\C_p| \cdot  \left( (1+\eps)\eps^{-1}\right) ^{2|S_j|},\end{align*}
	where the last inequality is because $QS(j,t) \subseteq S_j$ for any $t$. 
\end{proof}

\subsubsection{The congestion of an edge}
We are ready to prove our main result of the section, that the congestion of any edge $(\sigsig) \in F$ under the flow defined by the canonical paths from proper coloring to proper colorings, is polynomial in the number of vertices.

\begin{lemma}\label{lem:cong}
	The congestion of any transition $t$ under $f$, when $f$ is restricted to flows from proper colorings to proper colorings, satisfies
	$$\rho_f(t) \leq 2k(\lambda+1)n^3\log{n}   \left( (1+\eps)\eps^{-1}\right) ^{2\pw(G)}.$$
\end{lemma}

\begin{proof}
	From the definition of the congestion on an edge (Equation~\eqref{eq:congestion}), we have 
	
	\begin{subequations}
		\begin{align}
		\rho_f(t) &= \frac{1}{q(t)} \sum_{\cc \in \Omega} \sum_{p:t \in p \in \paths(\cc)} f(p)|p| \notag\\
		&\leq \frac{(\lambda+1)n\log{n}}{q(t)} \sum_{\cc \in \Omega} \sum_{p:t \in p \in \paths(\cc)} f(p)\label{eq1}\\
		&= 2|\Omega|k(\lambda+1)n^2\log{n} \sum_{\cc \in \Omega} \sum_{p:t \in p \in \paths(\cc)} f(p) \label{eq2}\\
		&= 2|\Omega|k(\lambda+1)n^2\log{n}\sum_{j =1}^n \sum_{\cc \in \Omega} f_{j,t,\ca,\cb}\label{eq3}\\
		&= 2|\Omega|k(\lambda+1)n^2\log{n}\sum_{j =1}^n  f_{j,t} \notag \\
		&\leq 2|\Omega|k(\lambda+1)n^2\log{n}\sum_{j =1}^n \pi(\cdot)^2|\C_p| \cdot  \left( (1+\eps)\eps^{-1}\right) ^{2|S_j|} \label{eq5}\\
		&= \frac{2 |\C_p|k(\lambda+1)n^2\log{n}}{|\Omega|} \sum_{j =1}^n  \left( (1+\eps)\eps^{-1}\right) ^{2|S_j|} \notag\\
		&\leq 2k(\lambda+1)n^3\log{n}   \left( (1+\eps)\eps^{-1}\right) ^{2\pw(G)}.\notag
		\end{align}
	\end{subequations}
	Inequality~\eqref{eq1} is because the length of any canonical path is at most $(\lambda+1)n\log{n}$; Equality~\eqref{eq2} is due to the definition of $q$: $q(\sigma,\sigma') = \pi(\sigma)P(\sigma, \sigma')$, where $\pi(\sigma)=|\Omega|^{-1}$ and $P(\sigma, \sigma')= (2kn)^{-1}$;
	Equality~\eqref{eq3} is simply a rephrasing that holds  because	$$\sum_{\cc \in \Omega} \sum_{p:t \in p \in \paths(\cc)} f(p)$$ is  the flow through $t$ under $f$; Inequality~\eqref{eq5} is due to Lemma~\ref{lem:fl}. The final inequality is due to Theorem~\ref{Kinner}, as the pathwidth of a graph equals its vertex separation number.
\end{proof}

\subsection{Mixing time}\label{sec:genflows}
Lemma~\ref{lem:cong} applies to the congestion from flow between proper colorings only. We extend this result to all of $f$. We rephrase our main lemma:

\begin{mainlemma}
	The congestion of any transition $t$ under $f$ satisfies
	$$\rho_f(t) \leq 8k^3(\lambda+1)n^5\log{n}   \left( (1+\eps)\eps^{-1}\right) ^{2\pw(G)}.$$
\end{mainlemma}

\begin{proof}
	
	We use the function $g'$ from singly-flawed to proper colorings described in Corollary~\ref{cor:sur} to define the flows that have a singly-flawed coloring as (at least) one of their endpoints. 
	For every $\cc$ such that  $\ca \in \C_{sf}$ and $\cb \in \C_p$, we route the entire flow on the transition  $(\ca, g'(\ca))$ and then proceed using the canonical paths described above for routing the flow from $g'(\ca)$ to $\cb$. If   $\ca \in \C_{p}$, and $\cb \in \C_{sf}$, we route the flow from $\ca$ to $g'(\cb)$ using the canonical paths above and then on the edge $(g'(\cb), \cb)$. Finally, if $\cc \in \C_{sf}$, we route the entire flow on  $(\ca, g'(\ca))$, use the canonical paths above to route from $g'(\ca)$ to $g'(\cb)$ and finally route the entire flow on $(g'(\cb), \cb)$. 
	For every state $\sigma \in \C_p$, there are at most $kn$ states $\sigma' \in\C_{sf}:g'(\sigma')=\sigma$. Therefore, we have multiplied the flow on every edge by at most \begin{equation}\label{g}
	k^2n^2+2kn+1< 4k^2n^2,
	\end{equation} where the first term is for pairs $\cc \in \C_{sf}$, the third is for $\cc \in \C_p$, and the second term on the left hand side is for mixed pairs. We added a further 
	\begin{equation}\label{eqj}
	kn \pi(\cdot)^2|\Omega| = \frac{2kn}{|\Omega|}<1\end{equation} to each edge  $(\sigma, g'(\sigma))$ and $(g'(\sigma), \sigma)$: there are at most $|\Omega|$ paths from a state $\sigma \in \C_p$ (to any other state), hence at most $kn|\Omega|$ paths from any $\sigma' \in \C_{sf}$. We absorb Inequality~\eqref{eqj} and the fact that $\pw(G) +\log{n} = (\lambda+1)\log{n}$  into Inequality~\eqref{g}. Multiplying the bound of Lemma~\ref{lem:cong} by $4k^2n^2$ gives the required bound.
\end{proof}

\section{The sampling algorithm}
In order to sample a proper coloring, we need to execute the Markov chain sufficiently many times to guarantee that w.h.p. it outputs a proper coloring, and when it does, return that coloring. The pseudo code is given as Algorithm~\ref{alg1}. For graphs of pathwidth bounded by $O(\log{n})$, the algorithm runs  time polynomial in $n$ and $\log{\delta}$, where $\delta$ is the required bias parameter.

\vspace{10pt}
\begin{algorithm}[H]\label{alg1}
	\caption{An almost-uniform sampler for proper colorings}\label{alg:mst}
	\SetKwInOut{Input}{Input}\SetKwInOut{Output}{Output}\SetKwInOut{Inquiry}{Inquiry}
	
	\Input{$G=(V,E)$ with maximal degree $\Delta$, a number of colors $k \geq \Delta+2$, a bias parameter $\delta>0$ }
	\Output{a proper $k$-coloring of $G$}

	\BlankLine
	Set $\eps=\lceil\frac{k}{\Delta}\rceil$\;
	Set $\delta_1 = \delta/(kn+1)^2$\;
	Set $T=\lceil \ln(3/\delta)(kn+2)^2\rceil$\;
	
	\For{$t=1$ to $T $}
	{
		Simulate $\mc(G,\eps)$ for $\tau(\delta_1)$ steps, starting from an arbitrary proper coloring\;
		If the final state $\sigma$ is a proper coloring, return $\sigma$; 
	}
	Return an arbitrary proper coloring\;
\end{algorithm}
\vspace{15pt}

\begin{theorem}
	Algorithm~\ref{alg1} is a fully polynomial almost uniform sampler for proper colorings with bias parameter $\delta$.
\end{theorem}
\begin{proof}
	We denote by $\hat{\pi}$  the distribution reached by $\mc$ after $\tau(\delta_1)$ steps. By definition, the total variation distance between and $\pi$ and $\hat{\pi}$  is at most $\delta_1$, hence  for any $S \subset \Omega$, it holds that
	\begin{equation}\label{eqrr} |\pi(S)-\hat{\pi}(S)| \leq \delta_1.\end{equation}
	Choosing $S=\C_p$  and applying Corollary~\ref{cor:same} gives that the probability of the final state being a proper coloring is at least $\frac{1}{kn+1}-\delta_1$. 
	Our choice of $T$ is so that Hoeffding's bound guarantees that Algorithm~\ref{alg1} will output a proper coloring during the {\bf for} loop (i.e., a final state of the Markov chain and not an arbitrary coloring), with probability at least $1-\delta_H$, where $\delta_H \leq \frac{\delta}{3}$. 
	To show that Algorithm~\ref{alg1} is an almost uniform sampler for proper colorings, we need to show that the sampled coloring is drawn from a distribution that is close to uniform. In other words, if $\sigma$ is the state output by Algorithm~\ref{alg1}, then for any $S \subseteq \C_p$
	$$\frac{\pi(S)}{\pi(\C_p)} - \delta \leq  \Pr[\sigma \in S] \leq \frac{\pi(S)}{\pi(\C_p)} + \delta.$$
	
	\begin{subequations}
		\begin{align}
		\Pr[\sigma \in S] &\geq \frac{\hat{\pi}(S)}{\hat{\pi}(\C_p)}(1-\delta_H)\label{q1}\\
		& \geq \frac{\hat{\pi}(S)}{\hat{\pi}(\C_p)}-\delta_H \notag\\
		& \geq \frac{\pi(S)-\delta_1}{\pi(\C_p)+\delta_1}-\delta_H \label{q2}\\
		& \geq \frac{\pi(S)-2\delta_1}{\pi(\C_p)}-\delta_H \label{q3}\\
		& =\frac{\pi(S)}{\pi(\C_p)}-\frac{2\delta_1}{\pi(\C_p)}-\delta_H \notag\\
		&\geq \frac{\pi(S)}{\pi(\C_p)}-\frac{2\delta}{3}-\frac{\delta}{3}, \notag
		\end{align}	 
	\end{subequations}
	where~\eqref{q1} is the probability that the Algorithm outputs a proper coloring and it is in $S$; ~\eqref{q2} is due to Equation~\eqref{eqrr}; ~\eqref{q3} is because $\frac{a-1}{b+1} \geq \frac{a+2}{b}$ for $b \geq a$.

	The complementary $\Pr[\sigma \in S] \leq \frac{\pi(S)}{\pi(\C_p)} + \delta$ is immediate by considering the set $\C_p \setminus S$: if this were not the case then it would hold that $\Pr[\sigma \in S]+\Pr[\sigma \in \C_p \setminus S] >1$.
	
\end{proof}

\paragraph{Acknowledgements} 
We thank Leonard Schulman for valuable discussions, and Adam Wierman and Leonard Schulman for comments on earlier versions of this draft.

\bibliographystyle{plain}\bibliography{refs}

\end{document}